\newif\ifieee
\ieeefalse

\documentclass[conference,letterpaper]{IEEEtran}
\usepackage[utf8]{inputenc}
\usepackage[T1]{fontenc}
\usepackage{url}
\usepackage{doi}
\usepackage{datetime}
\usepackage{hyperref}
\usepackage{cite}
\usepackage{amsfonts}
\usepackage{nicefrac}
\usepackage{tabularx}
\usepackage{array}
\usepackage[cmex10]{amsmath} % Use the [cmex10] option to ensure complicance
\usepackage{amssymb}
\usepackage{algorithm}
\usepackage{algpseudocode}

\usepackage{tikz}
\usepackage{circuitikz}
\usetikzlibrary{arrows,math,shapes,chains,calc,fit,positioning,%
	decorations.pathmorphing,automata}
\usepgflibrary{shapes.arrows}
\usepgflibrary{arrows.meta}

\newcolumntype{x}[1]{>{\centering\let\newline\\%
	\arraybackslash\hspace{0pt}}p{#1}}
\newcolumntype{y}[1]{>{\raggedright\let\newline\\%
	\arraybackslash\hspace{0pt}}p{#1}}
\newcolumntype{z}[1]{>{\raggedleft\let\newline\\%
	\arraybackslash\hspace{0pt}}p{#1}}

\newcommand\Tstrut{\rule{0pt}{2.5ex}}			% = `top' strut
\DeclareMathOperator\erf{erf}
\DeclareMathOperator\erfc{erfc}

\newtheorem{theorem}{Theorem}
\newtheorem{definition}{Definition}
\newtheorem{proof}{Proof}

\begin{document}
\title{On Entropy and Bit Patterns of Ring Oscillator Jitter}

%	Conference
\author{\IEEEauthorblockN{Markku-Juhani O. Saarinen}%
	\IEEEauthorblockA{PQShield Ltd., Oxford, UK\\
					mjos@pqshield.com}
}
%					Prama House,
%					267 Banbury Rd., Oxford OX2 7HT, UK\\

\maketitle

\begin{abstract}
	Thermal jitter (phase noise) from a free-running ring oscillator is 
	a common, easily implementable physical randomness source in True
	Random Number Generators (TRNGs). We show how to evaluate entropy,
	autocorrelation, and bit pattern distributions of ring oscillator
	noise sources, even with low jitter levels or some bias. Entropy
	justification is required in NIST 800-90B and AIS-31 testing and for
	applications such as the RISC-V entropy source extension. Our
	numerical evaluation algorithms outperform Monte Carlo simulations
	in speed and accuracy. We also propose a new lower bound estimation
	formula for the entropy of ring oscillator sources which applies more
	generally than previous ones.
\end{abstract}

\section{Introduction: Ring Oscillator Jitter}
\label{sec:intro}

	Free-running (ring) oscillators are widely used as physical noise
	sources in True Random Number Generators (TRNGs). In many ways, these
	designs are direct descendants of the oscillator-based
	``electronic roulette wheel'' used to generate the RAND tables of
	random digits in the late 1940s \cite{Br49}.

	A typical design (Fig. \ref{fig:ringosc}) has two oscillators; an 
	unsynchronized ring oscillator and a reference oscillator that is used
	to sample bits from the free-running oscillator. Spontaneous and
	naturally occurring phase shifts between the oscillators will cause
	unpredictability of output bits. These random oscillator period
	variations are known as oscillator jitter \cite{McRi09}.

	A pioneering  Ring Oscillator RNG chip was described and patented
	in 1984 by Bell Labs researchers \cite{FaMoCo84,CoFaMo84}.
	This type of noise source can be realized with
	``standard cells'' in HDL and requires no special manufacturing
	processes, making it a popular choice. More modern versions are used
	as noise sources for cryptographic key generation in common microchips
	from AMD \cite{AM17}  and ARM \cite{AR20}.

	Physical entropy sources are regulated in cryptographic
	security standards such as NIST's SP 800-90B \cite{TuBaKe:18}
	(for FIPS 140-3) and BSI's AIS 31 \cite{KiSc11} (for Common Criteria).
	These mandate health monitoring (built-in statistical tests) and
	appropriate post-processing. Cryptographic post-processing methods
	such as the SHA2 hash \cite{NI15A} completely mask statistical defects
	while still allowing guessing attacks. Noise source entropy
	evaluation is therefore crucial for determining the sampling rate and 
	``compression ratio'' of the conditioner.

	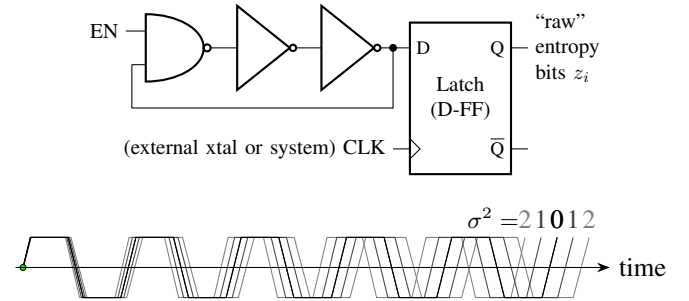
\begin{figure}[b]
	\centering
	%	daro.tikz
%	2020-11-28	Markku-Juhani O. Saarinen <mjos@pqshield.com>

\begin{circuitikz}[scale=0.8,transform shape]

	\draw	node[nand port,anchor=in 2](nand) at (0,0) {};
	\draw	node[not port,anchor=in](not1) at (nand.out) {};
	\draw	node[not port,anchor=in](not2) at (not1.out) {};
%	\draw	(not2.out) to[-*] ++(0,-1) -| (nand.in 2);

	\draw	(not2.out) node[circ] {} -- ++(down:1) -| (nand.in 2);

	\draw	node[flipflop D,anchor=pin 1,align=center] (dff) at (not2.out) 
		{{Latch}\\{(D-FF)}};
	\node[anchor=east,left] at (dff.pin 3) {(external xtal or system) {CLK}};

	\draw	(nand.in 1) -- ++(left:.1) node[anchor=east,left] {EN};

	\node[anchor=east,right,align=left](jit) at (dff.pin 6) 
		{{``raw''}\\{entropy}\\{bits $z_i$}};

%	\node[anchor=south,align=center] at ($ (not1) + (0,0.8) $) 
%		{\small\it ( 3,5,7.. inverters in a loop: continous oscillation. )};

\end{circuitikz}
	\vspace{2ex}
	%	haitari.tikz
%	2020-09-14	Markku-Juhani O. Saarinen <mjos@pqshield.com>

\begin{tikzpicture}[>=latex, scale=1.0, every node/.style = {scale=1.0}]

%\clip (-4.1,0.8) rectangle (+4.0,-2.1);

\coordinate (wav0) at (-3.9,-1.5);

\foreach \y in {2, 1, 0} {
	
	\foreach \s in {-1.0,+1.0} {

		\tikzmath{ \x=(1.0+\s*(0.03*\y)); \c=100-25*\y; }

		\draw[black!\c] (wav0)
		\foreach \i in { 1, 2, 3, 4, 5} {
			-- ++(\x*0.1,+0.4)
			-- ++(\x*0.5,+0.0)
			-- ++(\x*0.2,-0.8)
			-- ++(\x*0.5,+0.0)
			-- ++(\x*0.1,+0.4)
		}
		-- ++(\x*0.1,+0.4)
		node[above] {\y};
	}
}

\node[above] at  ($(wav0) + (6.25,0.4)$) {\small $\sigma^2=$};

\node [circle,draw=black,fill=green,inner sep=0.0,minimum size=0.5ex] 
	at (wav0) {};

\draw[-Stealth] ($(wav0) + (-0.1,0.0)$) -- ++(7.9,0.0) node[right] {time};

\end{tikzpicture}
	\caption{A ring oscillator consists of an odd  (here $N=3$) number of
		inverters connected into a free-running loop. The output is sampled
		using an independent reference clock, such as a crystal
		oscillator. Transition times are affected by jitter
		(largely from Johnson-Nyquist thermal noise), whose accumulation
		causes samples to become increasingly unpredictable.}
	\label{fig:ringosc}
	\end{figure}

\subsection{Physical Models and Their Limits}
\label{sec:physlim}

	An important contributor to the randomness of jitter in a
	ring-oscillator inverter loop (Fig. \ref{fig:ringosc}) is
	Johnson-Nyquist thermal noise \cite{Jo28,Ny28A}, which occurs
	spontaneously in any conductor (regardless of quality) as a result of
	thermal agitation of free electrons. Jitter is a macroscopic
	manifestation of this quantum-level \cite{CaWe51} Brownian
	effect.

	Timing jitter is a relatively well-understood phenomenon for many
	reasons. It is an important limiting factor to the synchronous
	operating frequency of any digital circuit.
	
	An example of a detailed physical model for ring oscillator phase
	noise and jitter is provided by Hajimiri et al.
	\cite{HaLe98,HaLe99,HaLiLe99}, which we recap here.
	The randomness of the timing jitter has a strongly Gaussian character.
	The jitter accumulates in the phase difference against the reference
	clock, with variance $\sigma_t^2$ growing almost linearly from one
	cycle to the next.

	Under common conditions, the transition length standard deviation
	(uncertainty) $\sigma_t$ after time $t$ can be estimated for
	CMOS ring oscillators as (after \cite[Eqns. 2.6,5.18]{HaLe99}):
	\begin{equation}
		\sigma_t^2 = \kappa^2 t \approx
		\frac{8}{3 \eta} \cdot
		 \frac{kT}{P} \cdot
		 \frac{V_{DD}}{V_\text{char}} \cdot t
		 \label{eqn:hajimiri}
	\end{equation}
	In this derivation of physical jitter $\kappa^2$ we note especially the
	Boltzmann constant $k$ and absolute temperature $T$; other variables
	include power dissipation $P$, supply voltage $V_{DD}$, device
	characteristic voltage $V_\text{char}$, and a proportionality constant
	$\eta \approx 1$. The number of stages ($N$) and frequency $f$ affect
	power $P$ via common dynamic (switching) power equations.

	As noted in \cite[Sect. 5.2.1]{HaLe99}, such derived models only
	express {\it ``inevitable noise sources''} -- not {\it ``extra
	disturbance, such as substrate and supply noise, or noise contributed
	by extra circuitry or asymmetry in the waveform''}	-- which will
	increase jitter. Many of these factors are difficult to model
	individually or are beyond digital designers' control.
	In practice $\kappa^2$ is measured experimentally, and the existence
	of jitter (and hence, fresh thermal noise entropy) is continuously
	monitored by auxiliary circuits that are a part of the TRNG.

\subsection{From Statistical Random Tests to Entropy Evaluation}
\label{sec:stattest}

	A 1948 report by RAND \cite{Br48} describes the statistical tests
	performed on the output of the ``million digits'' oscillator device
	\cite{RA55}. The tests were based on work by Kendall and Smith 
	\cite{KeBa38,KeBa39} with their late 1930s electromechanical random
	number device: Frequency test, Serial test, Poker test, and Gap test.
	It is remarkable that versions of these tests remained in use until
	the 2000s in the FIPS 140-2 standard \cite{NI01A}.
	
	While such ``black box'' statistical tests suites -- including Marsaglia's 
	DIEHARD and its successors \cite{Ma95,BrEdBa03} and NIST SP 800-22 
	\cite{RuSoNe:10} --- may be useful when evaluating pseudorandom generators
	for Monte Carlo simulations, they are poorly suited for security 
	applications. It is illustrative that a test existed in NIST
	SP 800-22 even in 2010 to see if an LFSR is {\it ``long enough''} to be 
	{\it ``considered random''}  \cite[Sect. 2.10]{RuSoNe:10}. Elementary
	cryptanalysis with finite field linear algebra shows that the internal
	state of an LFSR can be derived from a small amount of output,
	allowing both future and past outputs to be reproduced with little
	effort -- a devastating scenario if that output is to be used for
	cryptographic keying. 
	
	By 2001 at least the German AIS 20/31 \cite{KiSc11,ScKi02} Common
	Criteria IT Security evaluations had diverged from the purely black-box
	statistical approach and instead concentrated on quantifying entropy
	produced by a noise source, evaluation of its post-processing methods,
	and also considered implementation security, cryptanalytic attacks,
	and vulnerabilities. Current NIST security evaluation methodology of
	physical noise sources \cite{TuBaKe:18} also acknowledges that general
	statistical properties of raw noise are less important than
	evaluation of its entropy content, but at the time of writing, do not
	require stochastic models or detailed analysis of physical sources.
	
	For purposes of security engineering, pseudorandomness in the output 
	of the physical source is an unambiguously negative feature as it makes
	the assessment of true entropy more difficult. On the other hand,
	Redundancy from a well-behaved stochastic model is easily manageable via 
	cryptographic post-processing. Once seeded, standard (Cryptographic)
	Deterministic Random Bit Generators (DRBGs \cite{BaKe15}) guarantee
	indistinguishability from random, in addition to providing prediction
	and backtracking resistance.

\subsection{Ring Oscillators as Wiener Processes}
\label{sec:rowiener}

	Pioneering work on modern Physical RNG Entropy Estimation was presented
	by Killmann and Schindler\cite{KiSc08}, whose stochastic model uses
	independent and identically distributed transition times (half-periods)
	to model jitter. Baudet et al. \cite{BaLuMi:11} take a frequency domain
	(phase noise) approach.  Our model broadly follows these and also the
	one by Ma et al. \cite{MaLiCh:14}. 
	
	Baudet et al. propose a Shannon entropy lower bound 
	\cite[Eqn. 14]{BaLuMi:11}, which has been used in engineering
	(e.g. \cite{PeMuBo:16}):	%	\cite[Eqn. 1]{ZhChFa:18}
	\begin{equation}
		\label{eqn:baudeth}
		H_1 \geq 1 - \frac{4}{\pi^2 \ln 2} 
			e^{-4 \pi^2 Q}
		+ O\big( e^{-6\pi^2 Q} \big).
	\end{equation}
	Here $Q=\sigma^2 \Delta t$ (``quality factor'') corresponds to 
	$\kappa^2$ in the physical model (Eqn. \ref{eqn:hajimiri}).	
	We observe that the bound of Eqn. \ref{eqn:baudeth} is never lower than
	0.415 even when $Q$ approaches zero -- this estimate is safe to use only
	under some additional assumptions. 

%	b(q) = exp(-2*Pi^2*q)
%	hmin(q) = 1 - (4/(Pi^2*log(2))) * exp(-4*Pi^2*q)

\subsection{Our Goals: FIPS 140-3 and More Generic ROs}
\label{sec:introgoals}
	
	Prior works generally state that the frequency of the free-running
	oscillator is much higher than sampling frequency and that they do
	not have a harmonic relationship. The source is also often
	taken to be unbiased and assumed to have a relatively high amount of
	entropy per sample. In this work, we show how to compute entropy, 
	autocorrelation coefficients, and bit pattern probabilities also 
	for less ideal parameters. Our goal is to have guarantees for
	entropy and min-entropy in TRNG designs. This is required in current
	cryptography standards AIS 31 \cite{KiSc11} and
	FIPS 140-3 \cite{NICC21} / SP 800-90B \cite{TuBaKe:18} and for use
	in applications such as RISC-V Microprocessors \cite{SaNeMa21,_Ma21}.

\section{A Stochastic Model and its Distributions}
\label{sec:tmodel}
	
	We consider the jitter accumulation $\sigma^2 \sim Q$ at sample time 
	rather than the variance of (half) periods \cite[Sect. 2.2]{MaLiCh:14}.
	We also ease analysis by using the sampling period as a unit of time -- 
	sample $z_i$ is at ``time'' $i$, and variance is defined accordingly.
	Our time-phase accumulation matches
	with the physical model ($\kappa^2$ of Eqn. \ref{eqn:hajimiri}) and
	also accounts for spontaneous, purely Brownian transitions and ripple
	when the relative frequency $F$ of oscillators is very small or harmonic.
	
	For sampled digital oscillator sources, we may ignore the signal
	amplitude and consider a pulse wave with period $T$
	and relative pulse width (``duty cycle'') $D$. We assume a constant
	sampling rate and use the sample bits as a measure of time.

	We normalize the sinusoidal phase $\omega$ as
	$x = \frac{\omega - \delta}{2\pi}$ to range $0 \leq x < 1$, where 
	$\delta$ is the rising edge location. The average frequency
	$F \approx 1/T \bmod 1$ is a per-bit increment to the phase and
	$\sigma^2$ is its per-bit accumulated variance (Eqn. \ref{eqn:hajimiri}).

	\begin{definition}[Sampling Process]
	The behavior of a $(F,D,\sigma^2)$ noise source and its bit sampler is
	modeled as:
	\begin{align}
	\label{eqn:jitstep}
	x_{i}	& = \big( x_{i-1} + \mathcal{N}(F,\sigma^2) \big) ~ \bmod 1 \\
	\label{eqn:jitcases}
	z_{i}	& = \begin{cases}
					1 & \text{if } x_i < D,\\
					0 & \text{if } x_i \geq D. \\
				\end{cases}
	\end{align}
	Here $z_i \in \{0,1\}$ is an output bit, and $x_i \in [0,1)$ is the
	normalized phase at sampling time. $F$ is the frequency in relation to 
	the sampling frequency, and $\sigma^2$ represents jitter.
	\end{definition}
	
	Due to normalization ($x \bmod 1 \equiv x - \lfloor x \rfloor$), 
	and negative $-F$ symmetry, $F$ can be reduced to range 
	$[0,\nicefrac{1}{2}]$. One may view this as a ``harmonic''
	reduction but there is no restriction for the sampler to run 
	faster than the source oscillator.
	
	The sampling process can be easily implemented to generate simulated bits 
	$z_1, z_2, z_3, ..$ for given parameters $(F,D,\sigma^2)$.
	This Wiener process is clearly only an idealized stochastic
	model, and its applicability for modeling specific physical
	random number generators must be individually evaluated.

	\begin{figure}[tb]
	\centering
	\includegraphics[width=.45\textwidth]{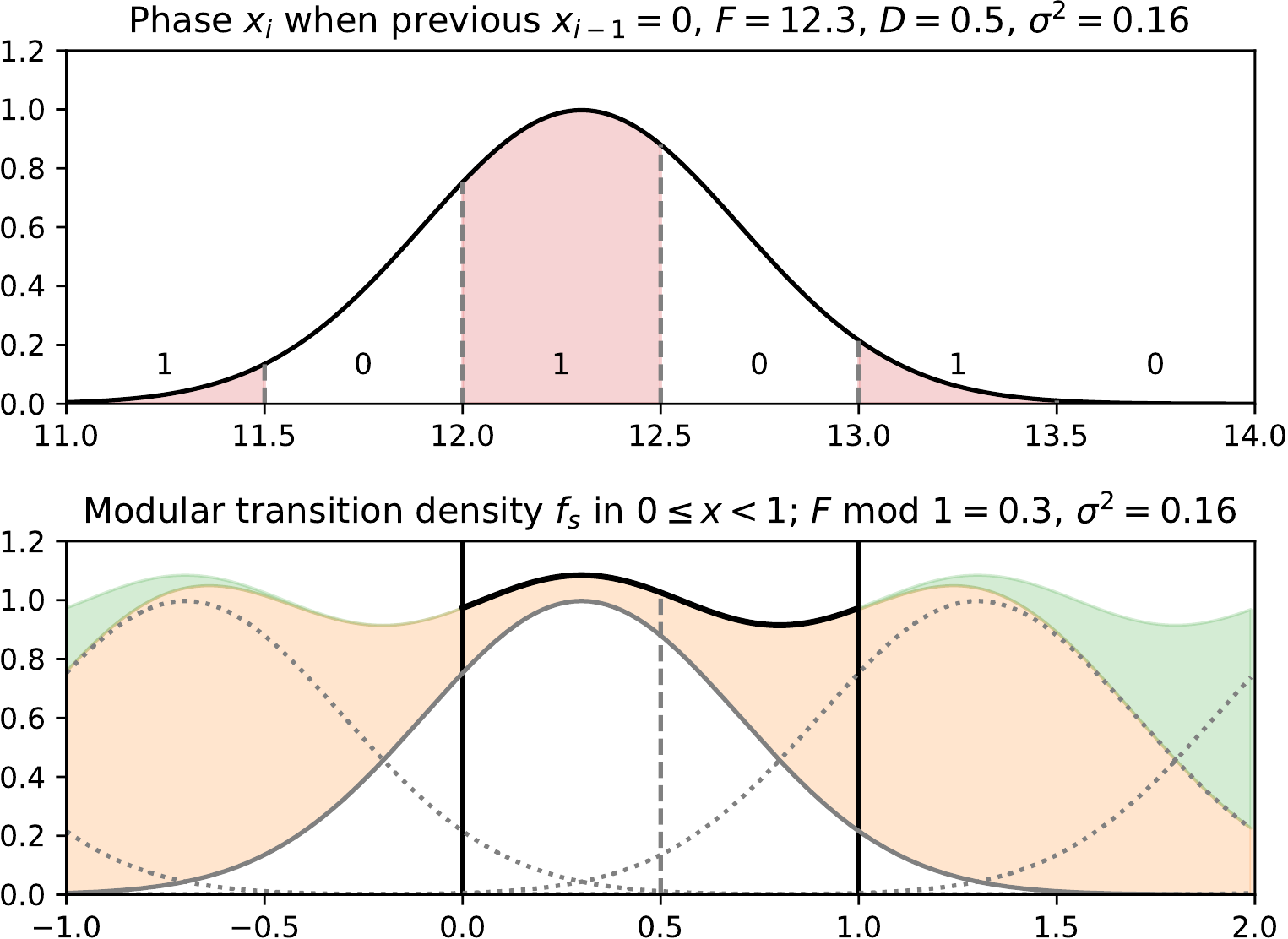}
	\caption{Gaussian phase transition and equivalent modular density.}
	\label{fig:stack}
	\end{figure}

\subsection{Distance to Uniform}
\label{sec:flatnesss}

	The Gaussian probability density function in Eqn. \ref{eqn:jitstep}
	becomes modularly wrapped (Fig. \ref{fig:stack}.) The classical
	assumption of ring oscillators is that if the accumulated variance
	$\sigma^2$ is large enough in relation to sampling rate, the modular
	step density function will become essentially ``flat'' in $[0,1)$;
	furthermore, if $(x_i-x_{i-1}) \bmod 1$ is uniformly random, then
	the bit sequence $z_i$ is correlation-free.	Some sources simply
	state ad hoc criteria for decorrelation
	(e.g. that $\sigma^2>1$).

	We will calculate the step function's statistical distance to the
	uniform distribution. 
	The density of the unbounded step function (Eqn. \ref{eqn:jitstep})
	can be equivalentl y defined
	over domain $0 \leq x < 1$ or as a 1-periodic function in 
	$\mathbb{R}/\mathbb{Z}$ 
	(See Fig. \ref{fig:stack}):
	\begin{equation}
		\label{eqn:psigma}
		f_s(x) = \frac{1}{\sqrt{2\pi\sigma^2}} \sum_{i \in \mathbb{Z}}
			e^{-\frac{(x - F + i)^2}{2\sigma^2}}.
	\end{equation}
	We have $f_s(a) = f_s(a+1)$ and
	$\int_a^{a+1} f_s(x)\,\mathrm{d}x = 1$ for all $a \in \mathbb{R}$.
	By choosing a tailcut value $\tau$ one can limit the sum to
	$\lfloor-\tau\sigma\rfloor \leq i \leq \lceil\tau\sigma\rceil$.
	This allows us to determine
	max at $f_s(F)$ and min at $f_s(F+\nicefrac{1}{2})$ for given $\sigma$.
	These are bounds for its statistical (total variation) distance to the
	uniform distribution (See Table \ref{tab:psigma}.) We see that this
	idealized ``1-dimensional lattice Gaussian'' would
	be cryptographically uniform at $\sigma^2 > 9$.

%	sd(x,sig) = l=ceil(50*ceil(sig));
%	sum(i=-l,l,exp(-((x+i)/sig)^2/2)) / (sig*sqrt(2*Pi))

	\begin{table}[htb]
	\centering
	\caption{Extrema of the probability density function $f_s(x)$ for
		some $\sigma$.	}
	\label{tab:psigma}
	\begin{tabular}{|x{4ex}| x{9ex} x{9ex}|c|x{4ex}|x{15ex}|}
		\cline{1-3} \cline{5-6} \Tstrut
		$\sigma$ & $f_s$ min & $f_s$ max & ~~
			& $\sigma$ & $f_s$ range \\
		\cline{1-3} \cline{5-6} \Tstrut
		0.10 & $0.000030$ & $3.989423$	&&	1.00 & $1 \pm 2^{-27.47766}$ \\
		0.20 & $0.175283$ & $1.994726$	&&	1.50 & $1 \pm 2^{-63.07473}$ \\
		0.30 & $0.663191$ & $1.340089$	&&	2.00 & $1 \pm 2^{-112.9106}$ \\
		0.50 & $0.985616$ & $1.014384$	&&	2.50 & $1 \pm 2^{-176.9854}$ \\
		0.75 & $0.999970$ & $1.000030$	&&	3.00 & $1 \pm 2^{-255.2989}$ \\
		\cline{1-3} \cline{5-6}
	\end{tabular}
	\end{table}

\subsection{Autocorrelation and Sampling Intervals}
\label{sec:bitac}
	We define a scaled, binary delay-$k$ autocorrelation measure 
	$-1 \leq C_k \leq +1$:
	\begin{equation}
		C_k = 2 \Pr(z_i = z_{i+k}) - 1.
		\label{eqn:ekdef}
	\end{equation}
	We may estimate $C_k$, $k \geq 1$ for a finite $m$-bit sequence as
	\begin{equation}
		C'_k = \frac{1}{m-k} \sum_{i=1}^{m-k} (2z_i-1)(2z_{i+k}-1).
		\label{eqn:ekestm}
	\end{equation}
	For convenience, we set $C'_0 = \frac{1}{m}\sum_{i=1}^{m} (2 z_i-1)$ to
	represent simple bias in the same vector; $C'_0$ approximates $2D-1$.

	\begin{theorem}
	\label{thm:jitkadd}
	With fixed $D$ and
	$\sigma^2 > 0$ or $F \not\in \mathbb{Q}$ we have
	\begin{equation}
		C_k(F,D,\sigma^2) = C_1(kF \bmod 1, D, k\sigma^2)
			\text{~for~} k \geq 1.
	\end{equation}
	\end{theorem}
	\begin{proof}
		The variance of independent random variables is additive by
		induction in $k$, as is the mean.
		The difference $x_k-x_0$ will then have the distribution
		$\mathcal{N}(kF,k\sigma^2) \bmod 1$.
		Only with either noisy or non-rational (non-harmonic) $F$
		we may take $x_0$ in Equation \ref{eqn:jitstep} to be uniformly
		distributed in $[0,1)$.
	\end{proof}

\subsection{Computing $C_k$ to High Precision Without Simulation}
\label{sec:compek}

	Let $p_{00}$, $p_{01}$, $p_{10}$, $p_{11}$ be frequencies
	of adjacent bit pairs $p_{({z_{i},z_{i+1}})}$ present in bit sequence
	$z_i$ (Eqn. \ref{eqn:jitcases}) in the model.

	We'll pick one, $p_{11} = \Pr(z_{i}=1 \text{ and } z_{i+1}=1)$. The
	condition $z_i = 1$ limits the density of $x_i$ to ``boxcar'' $g_1$:
	\begin{equation}
	\label{eqn:g1def}
		g_1(x) = \begin{cases}
				1 & \text{if } x \in [0, D) \\
				0 & \text{if } x \notin [0, D). \\
			\end{cases}
	\end{equation}
	We also define $g_0(x) = 1$ if $x \in [D, 1)$ and zero elsewhere.

	The addition of random variables corresponds to convolution of their
	density functions; convolution $f_1 = g_1 * f_s$ with the step function
	$f_s$ (Eqn. \ref{eqn:psigma}) yields the probability density of 
	$x_{i+1}$ conditioned on $x_i = 1$.
	The probability mass of the second bit $z_{i+1} = 1$ is in range
	$x_{i+1} \in [0,D)$ and we have
	\begin{equation}
		p_{11} = \int_0^D f_1(x) \,\mathrm{d}x.
		\label{eqn:p11int}
	\end{equation}

	Convolution $f_1 = g_1 * f_s$ density can be expressed as
	\begin{equation}
		\label{eqn:r01pdf}
		f_1(x) = \frac{1}{2} \sum_{i \in \mathbb{Z}}
			\big[ \erf(a_i) - \erf(b_i)	 \big]
	\end{equation}
	Where	$a_i = (x + i - F) / \sqrt{2\sigma^2}$ and
			$b_i = (x + i - F - D) / \sqrt{2\sigma^2}$.
	An indefinite integral $S_1=\int f_1(x)\,\mathrm{d}x$ with the same
	$a_i$,$b_i$ is
	\begin{equation}
		\label{eqn:integ}
		S_1(x)=\frac{\sqrt{2\sigma^2}}{2} \sum_{i \in \mathbb{Z}}
		\left[ a_i \erf(a_i) - b_i \erf(b_i)
			+ \frac{e^{-a_i^2} - e^{-b_i^2}}{\sqrt{\pi}} \right].
	\end{equation}
	Again, one can choose a tailcut bound $\tau$ for desired
	precision $\epsilon \approx \erfc(\tau/\sqrt{2})$ (via Gaussian CDF)
	and compute the sums just over the integer range
	$\lfloor-\tau\sigma\rfloor \leq i \leq \lceil\tau\sigma\rceil$.
	A typical choice for IEEE floating point is $\tau=10$
	(``ten sigma'').

	Choosing $p_{11} = \int_0^D f_1(x)\,\mathrm{d}x =
	S_1(D)-S_1(0)$ has some computational advantages.
	From $p_{11}$ we can derive other parameters
	$p_{01} = p_{10} = D - p_{11}$, $p_{00} = 1 - 2D + p_{11}$, and
	$C_1 = 4(p_{11}-D)+1$.
	To compute arbitrary $C_k$, substitute parameters $F' = kF \bmod 1$
	and $\sigma'^2 = k\sigma^2$ (Thm. \ref{thm:jitkadd}). We then have
	$C_k$ as $C'_1 = 4 [ S_1(D) - S_1(0) - D ] + 1$.

	Figure \ref{fig:chop4} shows the density functions $g$ for the four
	bit pair frequencies when $F=0.1$, $D=0.625$, $\sigma^2=0.04$
	($\sigma=0.2$).
	The dotted line on upper boxes corresponds to shape of $f_0$ and
	lower row to $f_1$; these have been chopped (shaded area) to
	$g_{00}, g_{01}, g_{10}, g_{11}$.
	Note that even though $g_{10}$ has a different
	shape from $g_{01}$, they have equivalent area and hence
	frequency $p_{01} = p_{10} = \frac{1-C_1}{4}$. This is natural since the
	frequency of rising edges must match the frequency of falling edges.

	\begin{figure}[tb]
	\centering
	\includegraphics[width=.45\textwidth]{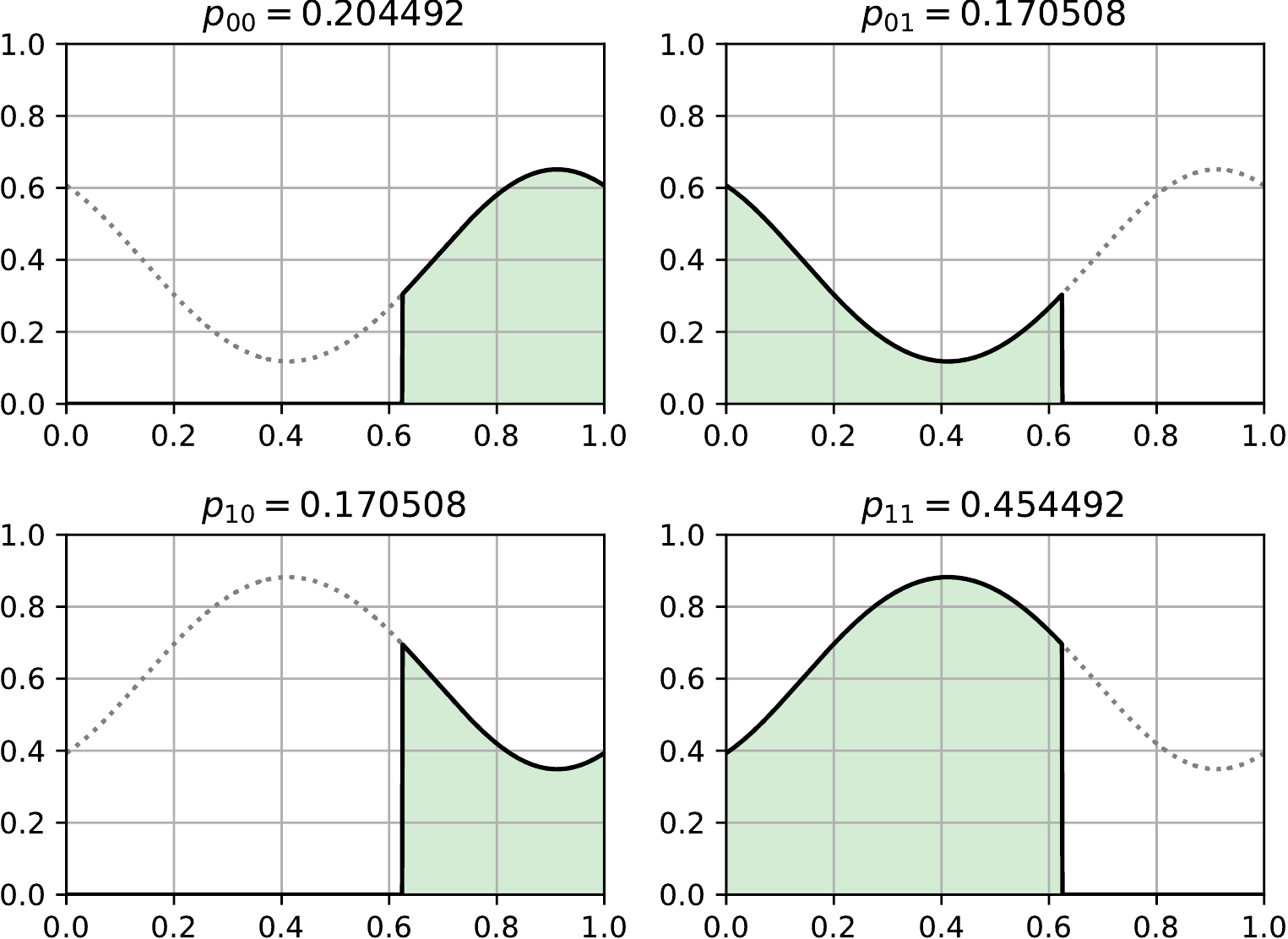}
	\caption{Bit pair probabilities for $F=0.1$, $D=0.625$, $\sigma^2=0.04$.}
	\label{fig:chop4}
	\end{figure}

\subsection{Use of $C_k$ in Modeling of Physical Sources}
\label{sec:physmod}

	The output from bit generation simulations agrees with
	these explicit autocorrelation values as expected.  
	Analytic $C_k$ is of course much faster to compute.
	
	Since autocorrelation estimates $C'_k$ (Eqn. \ref{eqn:ekestm})
	may also be easily derived from the output of physical ring oscillators,
	we can find a good approximate $(F,D,\sigma^2)$ model
	for a physical source by matching their autocorrelation properties. 
	We use least-squares minimization of few initial
	entries of autocorrelation vectors for this type of modeling.
	
	We can also experimentally derive parametrized models where frequency
	$F$ and jitter $\sigma^2$ are functions of environmental aspects such 
	as temperature or aging of circuitry; this, in turn, allows us to
	extrapolate and assign safe bounds for statistical health tests
	parameters and entropy output (yield of conditioner) over the lifetime
	of the device.

\subsection{Bit Pattern Probabilities via FFT Convolutions}
\label{sec:bitpat}

	To compute probabilities of bit triplets and beyond, we may ``chop'' a
	density function to zero the part which we know to be conditioned
	out; $g_{10}(x) = (f_1 \cdot g_0)(x)$, is $f_1$ chopped to zero
	outside $[D,1)$ range so we have 
	$\int_{-\infty}^\infty g_{10}(x) \,\mathrm{d}x = p_{10}$.

	Let $z$ be a
	sequence of bits for which conditional distribution $f_z$ is known;
	``chopping'' with  $g_0$ or $g_1$ and convoluting with step function
	$f_s$ we obtain distributions of one additional bit:
	$f_{z,0} = (f_z \cdot g_0) * f_s$ and
	$f_{z,1} = (f_z \cdot g_1) * f_s$.
	This chop-and-convolute process can be continued to determine the
	probability and phase distribution of an arbitrary bit pattern.
	
	Direct probability distribution integration formulas such as Eqn. 
	\ref{eqn:integ} become cumbersome for more generic bit
	patterns. We instead perform numeric computations on probability density
	functions $f_z$ represented as real-valued polynomial coefficients.
	This approach is attractive since the Fast Fourier Transform offers an
	especially efficient way to compute \emph{cyclic} convolutions of
	polynomials, as is required by our 1-periodic step function $f_s$ 
	(Eqn. \ref{eqn:psigma}). Unlike unbounded Gaussians our random
	variables $x_i \in [0,1)$ have a strictly limited range.

	\alglanguage{pseudocode}
	\begin{algorithm}
	\caption{Evaluate bit pattern probability $p_z$}
	\label{alg:pzfft}
	\begin{algorithmic}[1]
	\Function {pzfft}{$F$, $D$, $\sigma^2$, $z_1 z_2 \cdots z_n$}
	\For{$j\gets 0, 1, \cdots, m-1$}	\Comment{Init: Approximation.}
		\State $s_j \gets \frac{1}{m} f_s(\frac{j}{m})$
				\Comment{Eqn. \ref{eqn:psigma} for $F$ and $\sigma^2$.}
		\State $g_{1,j} \gets \max(\min(mD-j, 1), 0)$
				\Comment{Eqn. \ref{eqn:g1def} for $D$.}
		\State $g_{0,j} \gets 1 - g_{1,j}$	
				\Comment{Select zero -- inverse.}
		\State $v_j \gets \frac{1}{m}$
				\Comment{Start with uniform distribution.}
	\EndFor
	
	\For{$i\gets 1, 2, \cdots, n$}	\Comment{For each bit.}
		\For{$j\gets 0, 1, \cdots, m-1$}	\Comment{Chop half.}
			\State $t_j \gets v_j g_{z_i,j}$
				\Comment{Note bit select index $z_i$.}
		\EndFor
		\State{$v \gets t * s \bmod (x^m-1)$}
			\Comment{Convolution (FFT).}
	\EndFor
	
	\State \textbf{return} $p_z = \sum_{i=0}^{m-1} v_i$
		\Comment{Probability mass.}
	\EndFunction
	\end{algorithmic}
	\end{algorithm}

	These probability density functions $f(x)$ correspond to real-valued 
	$m$-degree polynomials $v=\sum_{i=0}^m x^i v_i$ in Algorithm 
	\ref{alg:pzfft}.
	The unit interval domain $x \in [0,1)$ is mapped to coefficients via
	$v_i \approx \int_{i/m}^{(i+1)/m} f(x) \,\mathrm{d}x$.
	For the step function of we approximate this 
	as $s_i = \frac{1}{m}f_s(i/m)$ and for chop functions so
	that $\sum_i g_{1,i} = D$ and $\sum_i g_{0,i} = 1-D$.
	We write the circular convolution using polynomial product
	and reduction modulo $x^m-1$, which can be very efficiently
	computed with FFT.

	The chopping operation is a point-by-point
	multiplication with $g_0$ or $g_1$ in the normal (time) domain, while
	step convolution is a point-by-point multiplication with $\hat{f_s}$ in
	the transformed (complex, frequency) domain, and hence each additional
	bit $z_i$ requires one forward and one inverse transform as $\hat{f_s}$ 
	remains the same.
	Our open-source, FFTW3-based \cite{FrJo05} portable C implementation
	allows accurate computation of probabilities of almost
	arbitrarily long patterns
	\footnote{Reference source code:
	\url{https://github.com/mjosaarinen/bitpat}}.

\section{Entropy Evaluation}
\label{sec:entropy}

	Let $Z_n$ be a random variable representing $n$-bit sequences 
	$z=(z_1,z_2,..z_n)$ which are sequentially generated by the 
	stochastic process of Sect. \ref{sec:tmodel} characterized
	by stationary parameters $(F,D,\sigma^2)$. Each of $2^n$ possible
	outcomes ${z}$ can be assigned a probability $p_z = \Pr(Z_n = z)$.

	The NIST SP 800-90B \cite{TuBaKe:18} entropy source standard focuses on
	min-entropy $H_\infty$, a member of the R{\'e}nyi family of entropies
	\cite{Re61}. Min-entropy (or ``worst-case entropy'')
	has a simple definition in case of a discrete variable, 
	based on the likelihood of the most likely outcome of $Z_n$:
	\begin{equation}
		\label{eqn:minentropy}
		H_\infty(Z_n) = \min_{z} ( -\log_2 p_z ) = -\log_2 ( \max_{z} p_z )
	\end{equation}
	The  AIS 31 \cite{KiSc11} Common Criteria evaluation method
	additionally uses the traditional Shannon entropy metric
	\begin{equation}
		\label{eqn:shannon}
		H_1(Z_n) = -\sum_z p_z \log_2 p_z.
	\end{equation}

	For Shannon entropy we consider its
	\emph{entropy rate} $H(Z)$. This is  a $[0,1]$-valued limit 
	$H(Z) = \lim_{n \to \infty} \frac{1}{n}H_1(Z_n)$.

\subsection{Entropy Upper Bounds}
\label{sec:hupbound}

	Probabilities $p_z$ obtained via Algorithm \ref{alg:pzfft} and
	related techniques in Sect. \ref{sec:bitpat}
	can be substituted to Eqns. \ref{eqn:minentropy} and
	\ref{eqn:shannon} to evaluate $H_\infty(Z_n)$ and $H_1(Z_n)$,
	respectively.
	
	Shannon entropy $H_1(Z_n)$ provides increasingly accurate upper bounds
	since we have 
	$H_\infty(Z_n) \leq H_1(Z_n)$ and
	\begin{equation}
	\label{eqn:hchain}
		H(Z) \leq \ldots \leq  \frac{1}{3}H_1(Z_3) 
		\leq \frac{1}{2}H_1(Z_2) \leq H_1(Z_1).
	\end{equation}
	This relationship follows from subadditivity of joint entropy 
	in case of Shannon entropy $H_1$; the monotonic relationship of Eqn. 
	\ref{eqn:hchain} does not hold for min-entropy $H_\infty$.
	
	All R{\'e}nyi entropies are upper bounded by max-entropy (Hartley entropy)
	$H_0$, i.e. the number (cardinality) of $n$-bit $z$
	with nonzero probability; $H_0(Z_n) = \log_2 |p_z > 0|$. If an
	$m$-bit encoding exists for all elements with $p_z > 0$ of $Z_n$,
	then its cardinality is at most $2^m$ and $H(Z) \leq H_0(Z_n) \leq m$.

	This leads to limit $H(Z) \to 0$ for a noiseless ($\sigma^2 = 0$)
	source, regardless of $F$ oscillation. 
	A simple $\epsilon,\delta$ argument shows that each $n$-bit
	sequence $z_i$ with $\sigma^2=0$ can be encoded by expressing
	$F,D$, and $x_0$ with asymptotic $O(\log n)$ bits of precision.
	Claim follows from $\lim_{n \to \infty} \frac{1}{n}{\log n} \to 0$.

\subsection{Entropy Lower Bounds as a function of $\sigma^2$}
\label{sec:hlobound}

	For an entropy lower bound we consider the entropy contribution 
	of jitter to an individual bit $z_i$ when all of the parameters
	$(F,D,\sigma^2)$ and the previous phase $x_{i-1}$ are known (in addition 
	to previous bit $z_{i-1}$). Let 
	$p_e = \Pr(z_i = z'_i)$ where the expected bit value
	$z'_i$ is deterministic (from $x_{i-1} + F$).
	
	We observe that $F$ cancels out in this case and we have 
	$p_e=p_{00}+p_{11}$ with $F=0$ for equations of Section \ref{sec:compek}.
	In case of an unbiased source $D=\frac{1}{2}$, a further simplification 
	yields frequency-independent bounds as a function of $\sigma^2$: 
	\begin{align}
		\label{eqn:pebound}
		p_e 		&= 2 \cdot [ S_1(\nicefrac{1}{2}) - S_1(0) ] 
					 = 4 \cdot S_1(\nicefrac{1}{2})	\\
		H_1(Z) 		&\geq -p_e \log_2 p_e -(1-p_e) \log_2 (1-p_e) \\
		\label{eqn:pemin}
		H_\infty(Z) &\sim -\log_2 p_e.
	\end{align}
	where $S_1$ is Eqn. \ref{eqn:integ} with $F=0$, $D=\frac{1}{2}$.
	From Eqn. \ref{eqn:pebound} we can show a looser approximate bound 
	$p_e \leq 1-\frac{1}{2}\tanh(\pi \sigma)$.

	These estimates are lower than some previously proposed
	lower bounds (See Eqn. \ref{eqn:baudeth}) as they are based on fewer
	assumptions. Crucially they cover the entire range of $\sigma^2$ -- 
	and are therefore safer to use in cryptographic engineering.

\subsection{Min-Entropy Estimates}
\label{sec:minentropy}

	One part of min-entropy estimation of $H_\infty(Z_n)$ is to find a
	maximum-likelihood bit sequence $z$, and the second is to determine
	its probability $p_z$. The second part can be accomplished with Algorithm
	\ref{alg:pzfft} -- we have $H_\infty(Z_n) = -\log_2 p_z$.

	A reasonable $z$ string ``guess'' is to select $x_0$ at random and use
	the peak probability path $x_i = x_{i-1} + F ~(\bmod~1)$ to
	determine $z_1,z_2,\cdots,z_n$. This approach is asymptotically
	sound, but overestimates entropy for small $n$.
	 
	A practical depth-first approach is to proceed as in Alg. \ref{alg:pzfft}
	but evaluate weights $q_0 = \sum_{j=0}^{m-1} v_j g_{0,j}$ and 
	$q_1 = \sum_{j=0}^{m-1} v_j g_{1,j}$ at each step $i$, and select
	$z_i$ with the higher probability mass.
	
	While $\max p_z$ can usually be found with subexponential $z$ guesses,
	worst-case complexity of this problem remains open. Certainly, a simple 
	depth-first search will not always work. Consider $\max p_z$ for source 
	$(D=\frac{1}{2}, F=0.15, \sigma^2=0.04)$:

%	\begin{center}
	\vspace{1ex}
	\begin{tabular}{l l}
	$\frac{1}{3}H_\infty(Z_3) = 0.844807$, 
		& $p_{000} = p_{111} = 0.172609$.	\\
	$\frac{1}{4}H_\infty(Z_4) = 0.849297$, 
		& $p_{0000} = p_{1111} = 0.0949171$.\\
	$\frac{1}{5}H_\infty(Z_5) = 0.846341$, 
		& $p_{00011} = p_{00111} = $ \\
		& $p_{11000} = p_{11100}=0.0532267$.
	\end{tabular}
	\vspace{1ex}
%	\end{center}

	We first note that the entropy increase from $Z_3$ violates subadditivity 
	(and would not be possible for $H_1$; Eqn. \ref{eqn:hchain}).
	Furthermore, the maximum-probability bit strings of $Z_4$ are not
	substrings of those for $Z_5$; not reachable via iteration.

\subsection{Comparison to SP 800-90B $H_\text{original}$ Estimation}

	Current NIST SP 800-90B methodology \cite[Section 3.1.3]{TuBaKe:18}
	suggests the use of stochastic models for the $H_\text{submitter}$
	analytic estimate. 
	A measurement-based black-box estimate $H_\text{original}$ also needs
	to be reported; for non-IID sources it is the
	minimum output of ten entropy estimation algorithms,
	for which an official software package is available from
	NIST\footnote{NIST:
	\url{https://github.com/usnistgov/SP800-90B_EntropyAssessment}}.
	
	We generated 16,000 simulated sequences of $8*10^6$ bits with random
	$\sigma$ and $F \in [0,\nicefrac{1}{2}]$, and subjected them to
	the official $H_\text{original}$ tests.
	Fig. \ref{fig:sp90b} contrasts these simulation results with an analytic 
	min-entropy estimate for $Z_{100}$ where $z$ is chosen to follow
	maximum probability mass (Section \ref{sec:minentropy}).
	
	As expected, the black-box heuristic which has been designed to 
	{\it ``lean toward a conservative underestimate of min-entropy''}
	\cite[Sect G.2]{TuBaKe:18} reports less entropy than our estimates.
	
	Fig. \ref{fig:sp90b} also shows $H_\infty(Z) \sim -\log_2 p_e$
	min-entropy derived from the  bit-prediction bound of Eqns.
	\ref{eqn:pebound} and \ref{eqn:pemin}. This curve mostly traces the
	lower reaches of the stochastic model estimates (which are scattered
	here due to randomness of $F$). We suggest that this is a safe
	min-entropy engineering estimate from variance 
	$\sigma^2$, assuming an unbiased source ($D=\nicefrac{1}{2}$).

	\begin{figure}[tb]
	\centering
	\includegraphics[width=.45\textwidth]{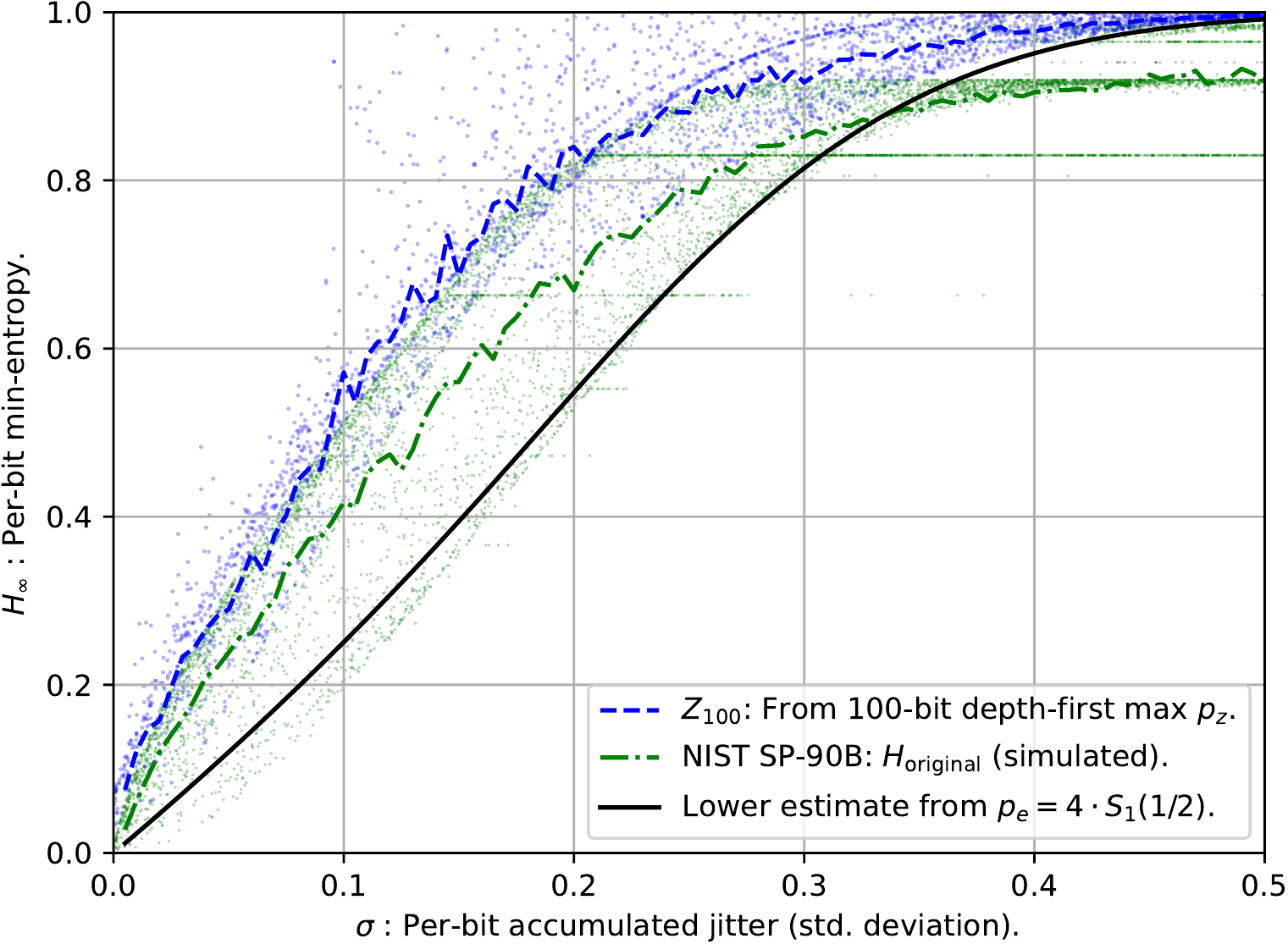}
	\caption{Min-entropy from distribution of $Z_{100}$ with
		depth-first $z$ selection,
		NIST SP 800-90B estimates, and our $p_e$ bit-prediction
		lower bound. Experimental data is represented as a scatter plot,
		with a line at the average.}
	\label{fig:sp90b}
	\end{figure}

\section*{Acknowledgments}

	\noindent
	Author thanks Joshua E. Hill for helpful comments. This work has been
	supported by Innovate UK Research Grant 105747.
%	References

\ifieee
	\newpage
	\bibliographystyle{IEEEtran}
\else
	\bibliographystyle{unsrturl}
\fi

\bibliography{entropy}

\end{document}